\newif\ifdraft \drafttrue \draftfalse
\newif\ifincludeappendix \includeappendixtrue 
\documentclass[sigconf, nonacm]{acmart}
\usepackage[utf8]{inputenc}
\usepackage{xargs,calc}
\usepackage{stmaryrd}
\usepackage{subfig}
\usepackage{adjustbox}
\usepackage{etoolbox}
\usepackage[commandnameprefix=always,commentmarkup=footnote]{changes}

\usepackage[inline]{enumitem}

\usepackage{tikz}
\usetikzlibrary{shapes.misc,tikzmark}
\usepackage{multirow}

\usepackage{svg,hyperref}
\hypersetup{%
  hidelinks
}

\usepackage[beginComment={//~},beginLComment={//~},endLComment={}]{algpseudocodex}
\tikzset{algpxIndentLine/.style={dotted}}

\newcounter{lastalgoline}
\newenvironment{vmalgorithm}%
{\begin{algorithmic}[1]\setcounter{ALG@line}{\value{lastalgoline}}}
{\setcounter{lastalgoline}{\value{ALG@line}}\end{algorithmic}}

\usepackage[scaled=0.8]{FiraMono}
\algrenewcommand\algorithmicindent{1.25em}
\algrenewcommand\alglinenumber[1]{\makebox[1em][r]{\footnotesize #1.}}
\algrenewcommand\algorithmicdo{}
\algrenewcommand\algorithmicprocedure{\textbf{\texttt{type}}}
\algrenewcommand\algorithmicthen{}

\algrenewcommand\algorithmicforall{\textbf{for each}}

\input{aux/macros}
\newlength{\thmspace}
\newtheoremstyle{vmstyle}%
    {\thmspace} 
    {\thmspace} 
    {\slshape\renewcommand{\emph}[1]{\textbf{#1}}} 
    {} 
    {\bfseries} 
    {.} 
    {2ex} 
    {} 
\theoremstyle{acmplain}

\newcommand{\thmBlockFont}[1]{#1}

\newcounter{thm}
\newtheorem{corollary}[thm]{\thmBlockFont{Corollary}}

\newtheorem{definition}[thm]{\thmBlockFont{Definition}}
\newtheorem{example}[thm]{\thmBlockFont{Example}}
\newtheorem{lemma}[thm]{\thmBlockFont{Lemma}}

\newtheorem{remark}[thm]{\thmBlockFont{Remark}}

\newtheorem{theorem}[thm]{\thmBlockFont{Theorem}}

\newtheorem{falsepropositionX}{\thmBlockFont{Proposition}}

\newtheorem{falsetheoremX}{\thmBlockFont{Theorem}}

\newtheorem{falsecorollaryX}{\thmBlockFont{Corollary}}

\newtheorem{falselemmaX}{\thmBlockFont{Lemma}}

\newtheorem*{falsestatementX}{\thmBlockFont{\thestatement}}
\newenvironment{falsestatement}[1]{\def\thestatement{#1}\begin{falsestatementX}}{\end{falsestatementX}\let\thestatement\relax}
\tikzset{
    annotation edge/.style={-,draw, dashed, rounded corners, shorten <=3pt, color=black!75},
    annotation/.style={draw,dashed, rectangle, rounded corners, inner sep=5pt, text=black,},
    llist/.style={draw,-stealth, in=0,out=0,looseness=4, thick},
    point/.style = {fill=black,minimum width=.5ex,circle,inner sep=0pt},
    point2/.style = {draw=fred,minimum width=2ex,circle,inner sep=0pt},
    vertex/.style = {draw, rounded rectangle},
    state/.style = {draw, circle, minimum width=5mm},
    label/.style = {rectangle, draw, fill = white, dashed, text = black},
    index/.style = {text = black},
    edge/.style = {draw,-stealth, shorten >=2pt},
    enum edge/.style = {draw,-stealth, thick, rounded corners, shorten >=5pt},
    transition/.style = {draw,-stealth, shorten >=2pt},
    high/.style = {edge,color = fblue},
    suspicious/.style = {edge, color=fred},
    high suspicious/.style = {edge, color = fpurple},
    abs label dir/.code args={#1}{\def\abslabeldirection{#1}},
    abs label sep/.code args={#1}{\def\abslabelsep{#1}},
    abs label/.style args={#1}{postaction={decorate, decoration={markings, mark=at position \abslabelsep with {\path node[anchor=\pgfdecoratedangle+\abslabeldirection] {#1};}}}},
    abs left label/.style args={#1}{absolute label dir=-90,absolute label={#1}},
    abs right label/.style args={#1}{absolute label dir=90,absolute label={#1}},
    vm loop/.style={transition, pos=.5, looseness = 9},
    north west loop/.style={vm loop, in={\the\numexpr 135 + \loopangle\relax}, 
                                     out ={\the\numexpr 135 - \loopangle\relax}},
    north east loop/.style={vm loop, in={\the\numexpr 45 + \loopangle\relax}, 
                                     out ={\the\numexpr 45 - \loopangle\relax}},
    south west loop/.style={vm loop, in={\the\numexpr -135 + \loopangle\relax}, 
                                     out ={\the\numexpr -135 - \loopangle\relax}},
    south east loop/.style={vm loop, in={\the\numexpr -45 + \loopangle\relax}, 
                                     out ={\the\numexpr -45 - \loopangle\relax}},
    north loop/.style={vm loop, in={\the\numexpr 90 + \loopangle\relax}, 
                                out ={\the\numexpr 90 - \loopangle\relax}},
    south loop/.style={vm loop, in={\the\numexpr 270 - \loopangle\relax}, 
                                out ={\the\numexpr 270 + \loopangle\relax}},
    east loop/.style={vm loop, in={\the\numexpr 0 + \loopangle\relax}, 
                                out ={\the\numexpr 0 - \loopangle\relax}},
    west loop/.style={vm loop, in={\the\numexpr 180 - \loopangle\relax}, 
                                out ={\the\numexpr 180 + \loopangle\relax}},     
}
\def\abslabeldirection{90}

\def\abslabelsep{5pt}
\def\loopangle{24}

\definecolor{fblue}{rgb}{0.,0.,0.75}%
\definecolor{fred}{rgb}{0.85, 0.1, 0.1}%
\definecolor{fpurple}{rgb}{0.6,0.2,0.4}%

\newcommand{\thetitle}{Distinct Shortest Walk Enumeration for RPQs}
\newcommand{\thesubtitle}{}

\title{\thetitle\ifdefempty{\thesubtitle}{}{\\[.25em] \Large \thesubtitle}}
\author{Claire David}
\orcid{0000-0002-3041-7013}
\affiliation{%
  \institution{LIGM, Univ. Gustave Eiffel, CNRS}
  \streetaddress{Bâtiment Copernic, 5 boulevard Descartes}
  \postcode{77420}
  \city{Champs-sur-Marne}
  \country{France}
}
\email{claire.david@univ-eiffel.fr}
  
\author{Nadime Francis}
\orcid{0009-0009-4531-7435}
\affiliation{%
  \institution{LIGM, Univ. Gustave Eiffel, CNRS}
  \streetaddress{Bâtiment Copernic, 5 boulevard Descartes}
  \postcode{77420}
  \city{Champs-sur-Marne}
  \country{France}
}
\email{nadime.francis@univ-eiffel.fr}

\author{Victor Marsault}
\orcid{0000-0002-2325-6004}
\affiliation{%
  \institution{LIGM, Univ. Gustave Eiffel, CNRS}
  \streetaddress{Bâtiment Copernic, 5 boulevard Descartes}
  \postcode{77420}
  \city{Champs-sur-Marne}
  \country{France}
}
\email{victor.marsault@univ-eiffel.fr}

\begin{abstract}
    We consider the Distinct Shortest Walks problem. Given two vertices $s$ and $t$ of a graph database $\Dc$ and a regular path query, we want to enumerate all walks of minimal length from $s$ to $t$ that carry a label that conforms to the query.

    Usual theoretical solutions turn out to be inefficient when applied to graph models that are closer to real-life systems, in particular because edges may carry multiple labels. Indeed, known algorithms may repeat the same answer exponentially many times.

    We propose an efficient algorithm for graph databases with multiple labels. The preprocessing runs in $\bigo{\card{\Dc}\times\card{\Ac}}$ and the delay between two consecutive outputs is in $\bigo{\lambda\times\card{\Ac}}$, where $\Ac$ is a nondeterministic automaton representing the query and $\lambda$ is the minimal length. The algorithm can handle $\varepsilon$-transitions in $\Ac$ or queries given as regular expressions at no additional cost.
\end{abstract}

\begin{document}

\maketitle

\begin{figure}[h]
    \resizebox{\linewidth}{!}{
    \begin{tikzpicture}[remember picture,-stealth]
    \node[vertex] (A) at (0,0) {$Alix$};
    \node[vertex] (B) at (9,0) {$Bob$};
    \node[vertex] (C) at (4.5,1.5) {$Cassie$};
    \node[vertex] (D) at (3,-1.5) {$Dan$};
    \node[vertex] (E) at (6,-1.5) {$Eve$};

    \draw (A) edge[in = 180, high] node[label] {$e_1: h$}   (C);
    \draw (A) edge[out = -45, in = 180, high suspicious] node[label] {$e_2:h,s$} (D);
    
    \draw (C) edge[high,out = 0] node[label] {$e_7:h$}  (B);
    
    \draw (D) edge[high] node[label] {$e_4:h$}  (E);
    \draw (C) edge[out = -90, in = 150, high] node[label,pos=0.6] {$e_5:h$}  (E);
    \draw (C) edge[out = -30, in = 90, suspicious] node[label, pos=0.35] {$e_6:s\vphantom{h}$}  (E);

    \draw (D) edge[out = 90, in = 210, suspicious] node[label] {$e_3:s\vphantom{h}$}  (C);

    \draw (E) edge[out = 0, in = 225, high suspicious] 
        node[label] {$e_8:h,s$} (B);
\end{tikzpicture}
    }
    \caption{A multi-edge multi-labeled graph database}
    \label{fig:bank}
\end{figure}

\section{Introduction}

Regular Path Queries (RPQs, \cite{CruzMendelzonWood1987, AnglesArenas+2017}) are the building block of most query languages over graph databases. Formally, an RPQ is defined by a regular expression~$R$ and is said to \emph{match} any walk in the database which carries a label that conforms to $R$. During query
processing, one fundamental task consists in producing all matching walks of minimal length that start and end at given vertices. In particular, this task is at the heart of the \emph{all-shortest-walks semantics}, one of the most widespread semantics in practice. For instance, it is the semantics of GSQL~\cite{GSQL3.9,DeutschXu+2019} (TigerGraph),
and of the theoretical language G-Core~\cite{AnglesArenas+2018}.
All-shortest-walks semantics is also supported by PGQL~\cite{PGQL2.0} (Oracle), and by GQL~\cite{GQL-ISO,FrancisGheerbrant+2023}. The latter is particularly relevant, as it has been designed to become the standard query language for graph databases.

In theoretical settings, this task is often considered to pose little challenge. Indeed, graph databases are usually abstracted as single-labeled graphs $\Dc$ (that is, graphs whose edges carry exactly one label from a finite alphabet), while queries are given as a deterministic finite automaton $\Ac$. In that case, finding shortest matching walks can be done as follows: construct the product graph $\Dc \times \Ac$, identify vertices that correspond to initial and final states of the automaton, discard all labels, and then simply run any well-known 
algorithm for finding shortest paths in an unlabeled graph.

However, this approach does not apply to real-life scenarios. 
Queries are typically given by the user as a regular expression, which does not translate to a deterministic automaton without a possible exponential increase in size. More importantly, real-life systems allow edges to carry multiple labels, either natively (as in GQL), or as a theoretical abstraction of boolean tests on data values. These two features lead to \emph{nondeterminism} both \emph{in the query} and \emph{in the data}. 
Thus, our goal is to efficiently solve the problem below.

\begin{problem}{Distinct Shortest Walks}
    \item[Inputs:] A multi-labeled multi-edge database~$\Dc$, and two vertices~$s,t$ in~$\Dc$.
    \item[Query:] A nondeterministic finite automaton~$\Ac$.
    \item[Output:] All shortest walks from $s$ to $t$ that match $\Ac$, without duplicates.
\end{problem}

Remark that this problem asks for a variable amount of outputs. In such settings, the standard approach consists in distinguishing between \emph{preprocessing} (time before the first output) and \emph{delay} (time between two consecutive outputs). This is known as \emph{enumeration complexity}; see \cite{Strozecki2021} for more details.

The main challenge is the handling of duplicates. Indeed, when either the query $\Ac$ or the database $\Dc$ allows nondeterminism, a single walk $w$ of $\Dc$ might correspond to exponentially many walks in $\Dc \times \Ac$. In that case, naively enumerating shortest walks in $\Dc \times \Ac$ would return an exponential number of copies of~$w$. 
One could ensure that each walk is returned only once by storing all outputted walks. In the worst case, this approach requires exponential space and leads to an exponential delay since the algorithm might find all copies of the same walk before discovering a new one.

Quite surprisingly, this problem has received little attention in the literature.
In \cite{MartensTrautner2018}, Martens and Trautner\footnote{Remark: Trautner is now known as Popp.} use a prior result due to Ackermann and Shallit \cite{AckermanShallit2009} to show that it can be enumerated with polynomial delay. 
The construction is given in more details in \cite{Popp2022}.
Note that although they consider the case of single-labeled databases, their result can be adapted to multi-labeled graph databases at no additional cost.
Proving precise complexity upper bounds
was not the main concern of \cite{MartensTrautner2018,Popp2022},
and \cite{AckermanShallit2009} did not use the enumeration complexity framework. Thus, the resulting algorithm only achieves a polynomial delay bound. 
The technical report \cite{FrancisMarsault2023} translates \cite{AckermanShallit2009} into the enumeration complexity framework, which leads to the following, more precise, statement.

\begin{theorem}[\cite{MartensTrautner2018}, combined with \cite{FrancisMarsault2023}]
\label{t:Martens-Trautner}
Given a nondeterministic automaton $\Ac$ with set of states $Q$ and transition table $\Delta$ and a database $\Dc$ with set of vertices $V$, \problemfont{Distinct Shortest Walks}($\Dc$,$\Ac$) can be enumerated with delay in $\bigo{\card{\Dc}\times\card{\Delta}
\times\lambda}$ after a preprocessing in $\bigo{\card[2]{Q}\times \card[2]{V}\times\lambda+\card{\Delta}\times\card{\Dc}\times\lambda}$, where $\lambda$ is the length of a shortest walk.
\end{theorem}

Other recent articles \cite{MartensNiewerth+2023, FariasRojasVrgoc2023} have been devoted to query evaluation under all-shortest-walks semantics, which once again highlights its significance as a core task in graph data management. However, in both cases, graphs are assumed to be single-labeled and the query is assumed to be given as an unambiguous automaton. These assumptions prevent all form of nondeterminism from the product graph $\Dc \times \Ac$ in the following sense: if a walk $w$ of $\Dc$ matches $\Ac$, there is only one witness of this fact in $\Dc \times \Ac$. In that case, the usual approach produces no duplicates. This leads to an algorithm with $\bigo{\card{\Dc}\times\card{\Ac}}$ preprocessing and $\bigo{\lambda}$ delay. 

\paragraph{Contributions.} We propose an efficient algorithm for \problemfont{Distinct Shortest Walks}. As stated in Theorem~\ref{theorem:main} below, the algorithm yields very satisfactory complexity bounds. In particular, the delay between two consecutive outputs during the enumeration phase \emph{does not depend} on the size of the input database, and the preprocessing phase is only linear in the size of the database.

\begin{theorem}\label{theorem:main}
    \problemfont{Distinct Shortest Walks} can be enumerated with a preprocessing time in $\bigo{\card{\Dc}\times\card{\Ac}}$ and a delay in $\bigo{\lambda\times \card{\Ac}}$, where~$\lambda$ is the length of a shortest walk.
\end{theorem}

Our algorithm has an additional $\card{\Ac}$ factor in the delay when compared to the simpler setting, as an extra cost for handling nondeterminism in the query and in the data. Note that it takes linear time to check whether a given automaton $\Ac$ is deterministic and a given database $\Dc$ is single-labeled. Thus, detecting that the input lies in the more favourable setting and running the more efficient algorithm instead can be done at no additional cost.

The main idea behind the algorithm is as follows. The enumeration phase consists in a depth-first traversal of the set of answers, represented as a backward-search tree rooted at the target~$t$.
Since all branches of $\Tc$ have the same length $\lambda$, this ensures the required bound on the delay between two consecutive outputs. Additionally, all walks are represented at most once in $\Tc$, which ensures that no output is returned twice. Ideally, the preprocessing phase would explicitly compute $\Tc$. This, however, can take an exponential time and space in the worst case scenario. To circumvent this issue, the algorithm will instead annotate the database with a limited amount of information that allows recomputing $\Tc$ \emph{on the fly} during the enumeration phase. This annotation makes use of carefully chosen data structures to ensure that recomputing the next edge of $\Tc$ \emph{does not depend} on the size\footnote{It does not even depend on the degree of the current vertex.} of $\Dc$ and is only linear in $\Ac$.

\paragraph{Outline.} Most of the article is devoted to describing the algorithm and proving its properties. After necessary preliminaries, Section~\ref{sec:algorithm} gives the algorithm in pseudocode and introduces the necessary tools to prove its correctness. Section~\ref{sec:complexity} establishes the complexity bounds that were claimed in Theorem~\ref{theorem:main}. Finally, Section~\ref{sec:extensions} discusses several extensions of the algorithm that do not affect its complexity. In particular, we show that the algorithm can easily be adapted in case the query is given as a regular expression or as an automaton with spontaneous $\varepsilon$-transitions. 

\section{Preliminaries}
\label{sec:preliminaries}

\subsection{Sets, lists and queues}

First, we explicitly state our assumptions about the data structures that are used to represent collections of elements. Indeed, the efficiency of the algorithm hinges in part on carefully choosing the data structure that is used at each step, as some structures are more efficient for traversal, ordered insertion, or copy. In particular, we use the following structures:

\subsubsection*{Maps and sets} 
A \emph{map} $S : X \rightarrow Y$ is implemented as an array of elements of $Y \cup \set{\bot}$ of size~$\card{X}$ hence:
    \begin{itemize}
        \item Creating and initializing a new map takes time $\bigo{\card{X}}$.
        \item Assigning an image $y$ to $x$ or removing $x$ takes time $\bigo{1}$.
        This is denoted as $S[x] \gets y$ or $S[x] \gets\bot$.
        \item Browsing a map takes time $\bigo{\card{X}}$.
    \end{itemize}
A \emph{set} over domain $X$ is simply a map : $X \rightarrow \set{\top,\bot}$.

\subsubsection*{Lists} 
Our lists can be modified in only one way: append an element at the head.  In particular, elements cannot be deleted or replaced. Such immutable singly-linked lists enjoy the following operations.
    \begin{itemize}
        \item Creating a new empty list takes time $\bigo{1}$.
        \item Appending in the head takes time $\bigo{1}$.
        \item Copying a list amounts to copying the head pointer, hence takes time $\bigo{1}$.
    \end{itemize}

\subsubsection*{Restartable queues} 
Our algorithm requires queues that may be restarted.
They are implemented as linked lists with three pointers: start, end, and current. Queues feature the following operations, which all take time $\bigo{1}$:
\begin{itemize}
    \item Creation of an empty queue.
    \item Enqueue, that adds an element at the end of the queue.
    \item Advance, that moves current to the next element.
    \item Peek, that retrieves the element pointed by current.
    \item Restart, that moves current to the start.
\end{itemize}

\subsection{Graph databases}
 
In this document, we model graph databases as multi-labeled, multi-edge directed graphs, and simply refer to them as \emph{databases} for short. Databases are formally defined as follows.

\begin{definition}
    A \emph{database} $\Dc$ is a tuple $(\Sigma, V, E,\src,\tgt,\lbl)$ where:%
    \begin{itemize*}[label={}]
        \item $\Sigma$ is a finite set of symbols, or \emph{labels};
        \item $V$ is a finite set of \emph{vertices};
        \item $E$ is a finite set of \emph{edges};
        \item ${\src}:E\rightarrow V$ is the \emph{source} function;
        \item ${\tgt}:E\rightarrow V$ is the \emph{target} function; and
        \item ${\lbl}:E\rightarrow 2^\Sigma$ is the \emph{labelling} function.
    \end{itemize*}
\end{definition}

Since this article establishes precise
complexity bounds, we need to make our assumptions about the memory representation of databases explicit. We assume the following:

\begin{itemize}    
    \item Every label takes space $\bigo{1}$ and equality of two labels can be checked in time $\bigo{1}$.
    \item 
    Every vertex~$v$ provides the following  in time $\bigo{1}$.
    \begin{itemize}
        \item $\incoming(v)$ : an array of pointers to the edges ending in~$v$
        \item $\indeg(v)$ : the in-degree of $v$, that is $\card{\incoming(v)}$
        \item $\outgoing(v)$ : an array of pointers to the edges starting in~$v$
        \item $\outdeg(v)$ : the out-degree of $v$, that is $\card{\outgoing(v)}$
    \end{itemize}

    \item
    Every edge~$e$ provides the following  in time $\bigo{1}$.
    \begin{itemize}
        \item $\src(e)$ : a pointer to the source vertex of~$e$
        \item $\tgt(e)$ : a pointer to the target vertex of~$e$
        \item $\lbl(e)$ : the label set of~$e$. We assume that we may browse through $\lbl(e)$ in time $\bigo{\card{\lbl(e)}}$.
        \item $\tgtidx(e)$ : the position of $e$ in $\incoming(\tgt(e))$, that is $e=\incoming(\tgt(e))[\tgtidx(e)]$.
    \end{itemize}
\end{itemize}
Altogether, the size of $\Dc = (\Sigma, V, E,\src,\tgt,\lbl)$ satisfies:
\begin{equation*}
    \card{\Dc}  \in \bigo{\card{V}+\card{E}+\sum_{e\in E} \card{\lbl(e)}}
\end{equation*}

\begin{remark}
    The only nonstandard item in the memory representation of databases is the function $\tgtidx$.
    Note that it may be precomputed in time $\bigo{\card{V}+\card{E}}$ if it is not natively provided by the database. Thus, this assumption does not change the complexity bounds promised in Theorem~\ref{theorem:main}.
\end{remark}

\begin{definition}
    A \emph{walk} $w$ in a database $\Dc$ is a nonempty finite sequence of alternating vertices and edges of the form
    $w=\walk{v_0,e_0,v_1,\ldots,e_{k-1},v_k}$ where~$k\geq 0$, $v_0,\ldots,v_k\in V$, $e_0,\ldots,e_{k-1}\in E$, such that:
    \begin{equation*}
        \forall i, 0\mathbin\leq i \mathbin< k, \quad \src(e_i) = v_i \quad\text{and}\quad \tgt(e_i)=v_{i+1}
    \end{equation*}

    We call~$k$ the \emph{length} of~$w$ and denote it by~$\len(w)$.
    We extend the functions $\src$, $\tgt$ and $\lbl$ to the walks in~$\Dc$ as follows.
    For each walk~$w=\walk{v_0,e_0,v_1,\ldots,e_{k-1},v_k}$ in~$D$,
    $\src(w) = v_0$, $\tgt(w)=v_k$, and 
    $\lbl(w)= \setst[\big]{a_0a_1\cdots a_{k-1}}{\forall i, 0\mathbin\leq i < k,~a_i\in\lbl(e_i)}$.
    Finally,~$s\xrightarrow{w}t$ means that $\src(w) = s$ and $\tgt(w) = t$.
    
    We say that two walks~$w,w'$ \emph{concatenate} if $\tgt(w)=\src(w')$, in which case we define their \emph{concatenation} as usual, and denote it by $w\cdot w'$, or simply~$ww'$ for short.

    For ease of notation, we will sometimes omit vertices from walks, as they are implicitly defined by the adjacent edges: for instance, in the database of Figure~\ref{fig:bank}, $\walk{e_1, e_7}$ is short for $\walk{Alix,e_1,Cassie,e_7,Bob}$. Similarly, for a walk $w$ and an edge $e$, we write $w\cdot e$ as a shorthand for $w\cdot \walk{\src(e),e,\tgt(e)}$.

    We use $\walks(\Dc)$ to denote the (possibly infinite) set of all walks in $\Dc$, and $\walks^{\leq\ell}(\Dc)$ to denote the restriction of $\walks(\Dc)$ to walks of length at most $\ell$.
\end{definition}

\subsection{Automata and queries}
\begin{definition}
    \label{def:automaton}
    A nondeterministic \emph{automaton} $\Ac$ is a 5-tuple $\aut{\Sigma,Q,\Delta,I,F}$ where~$\Sigma$ is a finite set of symbols,~$Q$ is a finite set of \emph{states},~$I\subseteq Q$ is the set of \emph{initial} states, $\Delta\subseteq{Q\times \Sigma\times Q}$ is the set of \emph{transitions} and~$F\subseteq Q$ is the set of \emph{final} states.

    As usual, we extend~$\Delta$ into a relation over $Q\times \Sigma^* \times Q$ as follows: for every~$q\in Q$, $(q,\varepsilon,q)\in\Delta$; and for every~$q,q',q''\in Q$ and every~$x,y\in \Sigma^*$, if~$(q,x,q')\in\Delta$ and~$(q',y,q'')\in\Delta$ then~$(q,xy,q'')\in\Delta$.
    We write~$\Delta(q,u)$ and~$\Delta^{-1}(u,q')$ as shorthands for the sets~$\Delta(q,u)=\{q'\mid(q,u,q')\in\Delta\}$ and~$\Delta^{-1}(u,q')=\{q\mid(q,u,q')\in\Delta\}$, respectively. 

    We denote by $L(\Ac)$ the \emph{language} of~$\Ac$, defined as follows.
    \begin{equation*}
        L(\Ac)= \setst[\big]{x\in\Sigma^*}{\exists i\in I,~\exists f\in F,~(i,x,f)\in\Delta}
    \end{equation*}
\end{definition}

As for databases, we explicitly state our assumptions about the memory representation of automata.

\begin{itemize}
    \item States are encoded as integers, that is $Q$ is assumed to be equal to~$\{0,\ldots,\card{Q}-1\}$.
    \item Given a state~$q$ and a label~$a$, one may access $\Delta(q,a)$ in time $\bigo{1}$, and $\Delta(q,a)$ is a (non-repeating) list of states.
\end{itemize}
Altogether, the size of~$\Ac$ is $\card{\Sigma}+\card{Q}+\card{\Delta}+\card{I}+\card{F}$, which is in $\bigo{\card{\Sigma}+\card{Q}+\card{\Delta}}$.

\begin{definition}
    A \emph{Regular Path Query} (RPQ) is defined by a regular language, given as an automaton $\Ac$.
    Given an RPQ $\Ac$ and a walk $w$ in a database $\Dc$, we say that $\Ac$ \emph{matches} $w$ (or equivalently, that $w$ matches $\Ac$) if $L(\Ac) \cap \lbl(w) \neq \emptyset$.
    We use $\match(\Ac,\Dc)$ to denote the (possibly infinite) set of walks that match $\Ac$.
\end{definition}

\subsection{Distinct shortest walks}

Given a query $\Ac$, a database $\Dc$, and two vertices $s$ and $t$ of $\Dc$, our goal is to enumerate the set of all shortest walks from $s$ to $t$ that match $\Ac$. 
This set is written $\sem{\Ac}(\Dc,s,t)$ and defined as follows.


\begin{definition}\label{def:lambda}
    Let $\Ac$ be an automaton, $\Dc$ a database and $s,t$ two vertices of $\Dc$. Let $\lambda = \min\setst{\len(w)}{w \in \match(\Ac,\Dc), s\xrightarrow{w}t}$. Then $\sem{\Ac}(\Dc,s,t)$ is defined as:
    \begin{equation*}
        \sem{\Ac}(\Dc,s,t) = \setst{w\in \match(\Ac,\Dc)}{\len(w) \mathbin= \lambda}
    \end{equation*}
\end{definition}

We can now restate the main problem as:

\begin{problem}{Distinct Shortest Walks}
    \item[Inputs:] A multi-labeled multi-edge database~$\Dc$, and two vertices~$s,t$ in~$\Dc$.
    \item[Query:] A nondeterministic finite automaton~$\Ac$.
    \item[Output:] Enumerate $\sem{\Ac}(\Dc,s,t)$, without duplicates.
\end{problem}

\begin{example}\label{ex:bank}
    In the graph database $\Dc$ of Figure~\ref{fig:bank}, vertices represent people and edges represent bank transfers. Transfers can have up to two labels: $h$ for ``high value'' and $s$ for ``suspicious''.

    Assume that we are searching for fraudulent behavior. We want to find sequences of transfers from Alix to Bob that contain only high value or suspicious transfers, with at least one of them being suspicious. This corresponds to computing $\sem{\Ac}(\Dc,Alix,Bob)$ with $\Ac$ being the two-state automaton that captures $h^* s(h+s)^*$.

    We can now remark the following:
    \begin{itemize}
        \item The shortest walk from Alix to Bob, $\walk{e_1, e_7}$, is of length 2. However, it does not match $\Ac$, as $hh\notin L(\Ac)$.
        \item Four walks of minimal length 3 match $\Ac$: $w_1 = \walk{e_1, e_5, e_8}$, $w_2 = \walk{e_1, e_6, e_8}$, $w_3 = \walk{e_2, e_3, e_7}$ and $w_4 = \walk{e_2, e_4, e_8}$.
        \item $w_4$ carries \emph{three} labels that belong to $L(\Ac)$: $shh$, $hhs$ and $shs$. It is still returned only \emph{once}. The same holds for $w_2$ and $w_3$, which carry \emph{two} suitable labels.
        \item Even though they visit the same vertices, $w_1$ and $w_2$ \emph{are not} the same walk, and are \emph{both} returned. Indeed, $e_5$ and $e_6$ do not represent the same transfer. In the database, they might have different amounts, dates, operating banks...
        \item The walk $w_5 = \walk{e_2, e_3, e_6, e_8}$ matches $\Ac$, but is not returned, as it is not of minimal length.
    \end{itemize}
\end{example}

\section{The algorithm}
\label{sec:algorithm}

\begin{figure*}[t]\newcommand{\unit}{\hspace*{1pt}}
    \bigskip
    \caption{Pseudocode of the main algorithm}
    \label{f:pseudocode}
    \bigskip
    \bigskip
    \bigskip
    \begin{minipage}[t][.88\textheight]{\linewidth/2-\columnsep/2}
    Inputs of the algorithm, also used as global variables:
    \begin{itemize}[labelindent=5pt,leftmargin=*]
        \item Automaton $\Ac = (\Sigma,Q,\Delta,I,F)$
        \item Database $\Dc = (\Sigma,V,E,\src,\tgt\lbl)$
        \item Source vertex $s\in V$
        \item Target vertex $t\in V$
    \end{itemize}

    \bigskip\bigskip\bigskip\bigskip
    
    \begin{vmalgorithm}
    \Function{Main}{\unit}
        \State $B, L, \lambda \gets \Call{Annotate}{\unit}$
        \State $C \gets \Call{Trim}{B}$
        \State $S \gets \{ q \mid  L_t\At{q} =\lambda \}$ \label{main:build first set}
        \State $\Call{Enumerate}{C,\lambda,\walk{t},S\cap F}$
    \EndFunction
    \end{vmalgorithm}

    \bigskip\bigskip\bigskip\bigskip

    \begin{vmalgorithm}
    \Function{Annotate}{\unit}
        \State $\ell \gets 0$   
        \State $\varname{current} \gets $ empty list
        \State $\varname{next} \gets $ empty list 
        \ForAll {vertex $u$ in $V$}
            \State $B_u \gets{}$new \Map: $Q \to \{ 0,\ldots, \indeg(u)-1 \}$
            \Statex \makebox[\linewidth][r]{$\to\List{Q}$}, 
            \Statex \phantom{$B_u \gets{}$}fully initialised with empty lists
            \label{annotate:init Au}
            \State $L_u \gets $ new empty \Map{} : $Q \to \N$
        \EndFor
            
        \ForAll {state $p$ in $I$}
            \State $L_s[p] \gets \ell$
            \State add $(s,p)$ to $\varname{next}$
        \EndFor

        \Statex

        \State $\varname{stop} \gets \bot$ \label{annotate:main loop}  
        \While {$\varname{next}$ is not empty ~\textbf{and}~ $\varname{stop}=\bot$}
            \State $ \ell \gets \ell+1$
            \State $\varname{current} \gets \varname{next}$
            \State $\varname{next} \gets $ empty list
            \ForAll {$(v,q)\in \varname{current}$}\smashedcommentpar{NB: $(v,q)\in \varname{current}$\\
            ~~$\iff L_v[q] = \ell-1$.}\label{annotate:for(v,q)}
                    \ForAll {$e\in Out(v)$}\label{annotate:for(e)}
                        \State $u \gets \tgt(e)$
                        \ForAll {$p \in \Delta(q,\lbl(e))$}\label{annotate:for(p)}
                            \If {$p \notin \dom(L_u)$}\label{annotate:inner loop start}
                                \smashedcommentpar{First time state $p$\\ is reached at vertex $u$}
                                \State $L_u[p] \gets \ell$
                                \State add $(u,p)$ to $\varname{next}$
                                \If{ $u=t$ ~\textbf{and}~ $p\in F$}                   \smashedcommentpar{First time a final\\ state is reached at\\ vertex $t$}
                                    \State $\varname{stop}\gets \top$ 
                                \EndIf
                                \Statex
                                \State add $q$ to $B_u[p][\tgtidx(e)]$
                                \label{annotate:update1}
                            \ElsIf {$L_u[p] = \ell$}
                                \smashedcommentpar{We found another walk of\\ length $\ell$ that reaches\\state $p$ at vertex~$u$. }
                                \Statex
                                \Statex
                                \State add $q$ to $B_u[p][\tgtidx(e)]$
                                \label{annotate:update2}\label{annotate:inner loop end}
                            \EndIf
                        \EndFor
                    \EndFor
            \EndFor
        \EndWhile 
        \State \Return {$(B, L, \ell)$}
    \EndFunction
    \end{vmalgorithm}%
    \end{minipage}\hfill%
    \begin{minipage}[t][.88\textheight]{\linewidth/2-\columnsep/2}
    \begin{vmalgorithm}
    \Function{Trim}{Annotation $B$}
        \ForAll{$u\in V$} \label{trim:for(u)}
            \State $C_u \gets{}$new Map $Q \to{} \Queue{In(u) \times \List{Q}}$, 
            \Statex $\phantom{C_u \gets{}}$initialised with empty queues
            \ForAll{$p\in Q$}\label{trim:for(q)}
                \ForAll{$e\in \incoming(u)$}\label{trim:for(e)}\smashedcommentpar{Recall that $\incoming(u)$ is sorted\\ in increasing $\tgtidx$ order.}
                    \Statex
                    \If{$B_u[p][\tgtidx(e)]$ is not empty}\label{trim:if not empty}
                        \State enqueue $\left(e,B_u[p][\tgtidx(e)]\right)$ in $C_u[p]$
                    \EndIf
                \EndFor
            \EndFor
            \Statex \commentpar[.9\linewidth]{Note that, for all~$u,p$, the queue $C_u[p]$ is sorted by first component: the pairs $(e,\_)$ are in increasing order of $\tgtidx(e)$}
        \EndFor
        \State \Return {$C$}
    \EndFunction
    \end{vmalgorithm}

    \bigskip\bigskip\bigskip\bigskip\bigskip

    \begin{vmalgorithm}
    \Function{Enumerate}{}
        (\begin{tabular}[t]{@{}l@{}}
         Trimmed annotation $C$,~
         Integer~$\ell$, 
         \\
         Walk~$w$,~
         State set~$S$)
        \end{tabular}
        \State $u \gets \src(w)$ \label{enumerate:beginning}
        \If {$\ell = 0$}\smashedcommentpar{Remark that $\ell = 0$ implies $u=s$ and~$S\subseteq I$. \\ No need to verify it.}
                \Statex
                \State \textbf{output $w$} \label{enumerate:base}
        \Else
            \While {$\top$} \label{enumerate:main loop}
                \State $e_{min} \gets nil$ \label{enumerate:compute min start}
                \ForAll {$p \in S$}
                    \If {$C_u[p]$ is not empty} 
                        \State $(e,X) \gets $ peek $C_u[p]$
                        \If {$e_{min} {=}\, nil$ \,\textbf{or}\, $\tgtidx(e) {<} \tgtidx(e_{min})$}
                            \State $e_{min} \gets e$
                        \EndIf
                    \EndIf
                \EndFor
                \If{$e_{min} = nil$} \smashedcommentpar{All queues are exhausted. Never \\ happens on the first iteration.}
                    \Statex
                    \ForAll{$p\in{}S$}
                        \State restart $C_u[p]$ \label{enumerate:restart}
                    \EndFor
                    \State \Return
                \EndIf \label{enumerate:compute min end} \label{enumerate:return}
                
                \Statex
                \State $S' \gets $ new empty subset of~$Q$ \label{enumerate:compute X start}
                \ForAll {$p \in S$}
                    \If {$C_u[p]$ is not empty} 
                        \State $(e,X) \gets $ peek $C_u[p]$
                        \If {$e = e_{min}$}
                            \ForAll{$q \in X$}\label{enumerate:for(q)}
                                \State $S'[q] \gets \top$
                            \EndFor
                            \State advance $C_u[p]$
                        \EndIf
                    \EndIf
                \EndFor \label{enumerate:compute X end}
    
                \State \Call{Enumerate}{$C$,~ $\ell\,{-}\,1$,~ $e_{min}\cdot w$,~ $S'$}
                \label{enumerate:enumerate}

            \EndWhile
        \EndIf
    \EndFunction
    \end{vmalgorithm}%
    \end{minipage}
\end{figure*}









\begin{figure*}
    \begin{tikzpicture}[remember picture, edge/.style = {-stealth, shorten >=5pt}]
    \node[vertex] (A) at (0,0) {$Alix$};
    \node[vertex] (B) at (12,0) {$Bob$};
    \node[vertex] (C) at (6,1.5) {$Cassie$};
    \node[vertex] (D) at (4,-1.5) {$Dan$};
    \node[vertex] (E) at (8,-1.5) {$Eve$};

    \draw (A) edge[in = 180, high] node[near start, above, sloped] {$h$} node[above, pos=.94,index] {1} (C);
    \draw (A) edge[out = -45, in = 180, high suspicious] node[midway, above, sloped] {$h,s$} node[pos=.9,above,index] {0} (D);
    
    \draw (C) edge[high,out = 0] node[near start, sloped, above] {$h$}  node[pos=.97,left,index,anchor=east, inner sep=5pt] {1} (B);
    
    \draw (D) edge[high, pos=0.9] node[near start, sloped, above] {$h$} node[below,index, pos=.9,anchor=north] {0} (E);
    \draw (C) edge[out = -90, in = 150, high] node[midway, sloped, above] {$h$} node[above,index,pos=.9,anchor=south] {1} (E);
    \draw (C) edge[out = -30, in = 90, suspicious] node[midway,sloped,above] {$s$} node[right,index, pos=.875, anchor=west] {2} (E);

    \draw (D) edge[out = 90, in = 210, suspicious] node[near start,sloped, above] {$s$} node[pos=.92, below,index] {0} (C);

    \draw (E) edge[out = 0, in = 225, high suspicious] 
        node[above,near start, sloped] {$h,s$} 
        node[pos=.95,left,index,anchor=east, inner sep=5pt] {0} (B);

    \newcommand{\annot}[2]{%
    \begin{tabular}{ccccl@{}}
    q&$L_{#1}[q]$&i&$B_{#1}[q][i]$&$C_{#1}[q]$\\\midrule #2\end{tabular}}
    \newcommand{\twolines}[1]{\multirow{2}{*}{#1}}
    \newcommand{\threelines}[1]{\multirow{3}{*}{#1}}
    \newcommand{\myrule}{\rule{0pt}{\heightof{0}}}
    \newcommand{\subnoderule}[1]{\subnode{#1}{\myrule}\hspace*{3em}\subnode{#1r}{\myrule}}

    \path[annotation edge] (A) -| ++(135:25mm) node[annotation, anchor=south] {
    \begin{tabular}{cc}
    q&$L_{Alix}[q]$\\
    \midrule
    0 & 0 \\
    \midrule
    1 & $\bot$
    \end{tabular}};

    \path[annotation edge] (B) to ++(90:2cm) node[annotation,anchor=-90] {\annot{Bob}{
        \twolines{0} & \twolines{2} &0 & [] & \\
        \cline{3-4}
                       & &1 & [0] & \subnoderule{B1}\\
        \midrule
        \twolines{1} & \twolines{3}  &0 & [1,0,1] & \subnoderule{B3} \\ 
        \cline{3-4}
                       & &1 & [1]   & \subnoderule{B4}\\
    }};
    
    \path[llist] (B1r) to[out=180] (B1) node[point]{};;
    \path[llist] (B3r) to[out=180] (B3) node[point]{} to (B4) node[point]{}; 

    \path[annotation edge] (C) to ++(90:1cm) node[annotation,anchor=-45] {\annot{Cassie}{
        \twolines{0} & \twolines{1} &0 & []   & \subnoderule{C1} \\ 
        \cline{3-4}
                       & &1 & [0]  & \subnoderule{C2} \\
        \midrule
        \twolines{1} & \twolines{2}  &0 & [0,1] & \subnoderule{C3}  \\  
        \cline{3-4}
                       & &1 & [] & \subnoderule{C4}  \\ 
    }};
    
    \path[llist] (C2r) to[out=180] (C2) node[point]{};
    \path[llist] (C3r) to[out=180] (C3) node[point]{};

    \path[annotation edge] (D) |- ++(-4mm,-15mm) node[annotation,anchor=east] (D-annot) {\annot{Dan}{
        0 & 1 & 0 & [0]  & \subnoderule{D1}\\
        \midrule
        1 & 1 & 0 & [0] & \subnoderule{D2}\\ 
    }};%
    
    \path[llist] (D1r) to[out=180] (D1) node[point]{};%
    \path[llist] (D2r) to[out=180] (D2) node[point]{}; %

    \path[annotation edge] (E) |- ++(8mm,-20mm) node[annotation,anchor=west] {\annot{Eve}{
        \threelines{0} & \threelines{2} &0 & [0]  & \subnoderule{E1}\\ 
        \cline{3-4}
                         & &1 & [0]  & \subnoderule{E2}\\  
        \cline{3-4}
                         & &2 & []   & \subnoderule{E3}\\
        \midrule
        \threelines{1} &\threelines{2}  &0 & [1] &\subnoderule{E4} \\ 
        \cline{3-4}
                         & &1 & []  &\subnoderule{E5}\\ 
        \cline{3-4}
                         & &2 & [0] &\subnoderule{E6}\\
    }};
    \path[llist]
        (E1r) to[out=180] 
        (E1) node[point]{} to 
        (E2) node[point]{};
    \path[llist] (E4r) 
                    to[out=180] (E4) node[point]{} 
                    to[looseness=2] (E6) node[point]{};

    \path (current bounding box.south) ++(-7mm,4mm) node[state] (s0) {0} ++(0:1.5cm) node[state, double] (s1) {1};

    \path(s0) ++(180:1.2cm) node (Ac) {$\mathcal{A}:$};
    \path(s0) ++(180:.75cm) node (s0i){};
        
    \path[transition] (s0) to node[above] {$s$} (s1);
    \path[transition, north loop, shorten >=2pt] (s0) to node[above] {$h$} (s0);
    \path[transition, north loop] (s1) to node[above] {$h+s$} (s1);
    \path[transition, shorten >=0] (s0i) to (s0);
\end{tikzpicture}%
    \caption{Automaton $\Ac$ for $h^*s(h+s)^*$ and the annotation of the example database $\Dc$ after preprocessing for $\sem{\Ac}(\Dc,Alix,Bob)$ }
    \label{fig:annotated}
\end{figure*}

In this section, we describe the main algorithm and introduce the necessary tools to prove its correctness. The pseudocode of the main function \functionfont{Main} and its subfunctions \functionfont{Annotate}, \functionfont{Trim} and \functionfont{Enumerate} are given in Figure~\ref{f:pseudocode}, page~\pageref{f:pseudocode}. This section provides the main ideas and formal statements of all steps of the proofs. Due to space constraints, the technical details of the proofs have been moved to Appendix~\ref{appendix:proofs}.

We fix a graph database $\Dc = (\Sigma, V, E,\src,\tgt,\lbl)$, a query $\Ac = \aut{\Sigma,Q,\Delta,I,F}$,  and two vertices $s$ and $t$ of $\Dc$. We assume that there exists at least one walk from $s$ to $t$ that matches $\Ac$ and we let $\lambda$ be the minimal length of any such walk.

The \functionfont{Main} function just calls the other three functions.
The first two functions, \functionfont{Annotate} and \functionfont{Trim}, correspond to
the preprocessing phase.  Figure~\ref{fig:annotated} gives the annotations $L_u$, $B_u$ and $C_u$ of the preprocessing for Example~\ref{ex:bank}. The $\tgtidx$ of edges are written at the head of arrows. 
The last function, \functionfont{Enumerate}, corresponds to the enumeration phase.
The following sections detail each function.

\subsection{Annotate}

\functionfont{Annotate} is the main preprocessing function. Intuitively, its purpose is to precompute, for all vertex $u$ of $\Dc$ and state $p\in Q$, the shortest walks $w$ that start in $s$, end in $u$ and carry a label that reaches $p$ in $\Ac$. However, representing this set of walks explicitly could take exponential time and space. Instead, \functionfont{Annotate} will annotate each vertex $u$ with two maps, $L_u$ and $B_u$.%
The map $L_u$ ("length") records the length of the shortest walks that can reach the states of $\Ac$ at $u$.  The map $B_u$ ("back") records the last edge used along said walks, in order to be able to reconstruct them backwards.

Formally, if $w$ is a shortest walk from $s$ to $u$ with a label that reaches $p$, it will be reflected as follows in the two maps:

\begin{itemize}
    \item $L_u[p] = \len(w)$. $L_u$ is a partial map: if no such $w$ exists, $L_u[p]$ is undefined.
    \item $q\in B_u[p][\tgtidx(e)]$, where $e$ is the last edge of $w$ and $p\in\Delta(q,\lbl(e))$, that is, $q$ is a possible predecessor state to $p$ along the walk. $B_u[p][\tgtidx(e)]$ is a list, and thus may contain duplicates. However, its total size will never exceed $\sum_{a\in\Sigma}\card{\Delta^{-1}(a,p)}$.
\end{itemize}

\noindent \functionfont{Annotate} populates the maps by performing a breadth-first traversal of the product graph $\Dc \times \Ac$.  It stops at the end of step $\lambda$, where $\lambda$ is the length of a shortest walk from $s$ to $t$ that matches $\Ac$. Indeed, it will be the first iteration of the traversal in which the target vertex $t$ is reached together with a final state of~$\Ac$.

The following lemma ensures the correctness of \functionfont{Annotate}. 

\begin{lemma} At the end of \functionfont{Annotate}, the following properties hold for all $p,q\in Q$, $u\in V, e\in E$ and $i\leq \indeg(u)$:
    \begin{enumerate}
        \item $L_u[p] = \min\setst*{\len(w)}{
            \begin{array}{l}
                w\in\walks^{\leq\lambda}(\Dc), s\xrightarrow{w}u \\ \text{ and } p\in\Delta(I,\lbl(w))
            \end{array}}$
        \label{annotatelemma:a}

        \item $q\in{}B_u[p][i]$ if and only if there exists a walk $w$ from $s$ to $u$ of the form $w = w'\cdot{}e$ such that:
            \begin{itemize}
                \item $\len(w) = L_u[p]$
                \item $\tgtidx(e) = i$
                \item $q\in\Delta(I,\lbl(w'))$
                \item $p\in\Delta(q,\lbl(e))$
            \end{itemize}
        \label{annotatelemma:b}

        \item $B_u[p][\tgtidx(e)]$ is of size at most $\sum_{a\in\Sigma}\card{\Delta^{-1}(a,p)}$.
        \label{annotatelemma:c}
    \end{enumerate}
    \label{lem:annotate}
\end{lemma}

\subsection{Trim}

\functionfont{Trim} is the second step of the preprocessing phase.  It translates each map $B_u[p]$
into a queue~$C_u[p]$, essentially by removing empty $B_u[p][\tgtidx(e)]$.
During the enumeration phase, this allows to efficiently iterate over all edges~$e$ such that the list $B_u[p][\tgtidx(e)]$ is nonempty.
Indeed, browsing $B_u[p]$ directly would increase the delay by a factor~$d$, the maximal in-degree of~$\Dc$.




Formally, $C_u[p]$ is a queue of pairs $(e,X)$, where $e\in\incoming(u)$ and $X$ is a nonempty list over $Q$. The correspondence between $C_u$ and $B_u$ is formally stated in the lemma below, whose proof immediately follows from the pseudocode of \functionfont{Trim}.

\begin{lemma}
    At the end of \functionfont{Trim}, the following properties hold for all $p\in Q$, 
    $u\in V$, $e,e'\in \incoming(u)$ and lists $X,X'$ over $Q$:
    \begin{enumerate}
        \item $(e,X) \in C_u[p]$ if and only if $X = B_u[p][\tgtidx(e)]$ and $X \neq \emptyset$.
        \item $C_u[p]$ is sorted in increasing $\tgtidx(e)$ order: if $(e,X)$ appears before $(e',X')$ in $C_u[p]$, then we necessarily have $\tgtidx(e) < \tgtidx(e')$. In particular, $e \neq e'$.
        \item If $(e,X) \in C_u[p]$, then $X$ is of size at most $\sum_{a\in\Sigma}\card{\Delta^{-1}(a,p)}$.
    \end{enumerate}
    \label{lem:trim}
\end{lemma}

\subsection{Enumerate}

Finally, \functionfont{Enumerate} makes use of the precomputed structures $C_u$ to handle the enumeration phase. In order to prove its correctness, we formally define the backward-search
tree of the set of answers and show how it relates to $C_u$.

\begin{definition}
    The backward-search 
    tree $\Tc$ of $\sem{A}(\Dc,s,t)$ is the tree defined as follows:
    \begin{enumerate}
        \item The nodes of $\Tc$ are the suffixes of walks in $\sem{A}(\Dc,s,t)$: a walk $w$ is a node of $\Tc$ if and only if there exists a walk $w'$ such that $w' \cdot w \in \sem{A}(\Dc,s,t)$.
        \item The root of $\Tc$ is the walk $\walk{t}$.
        \item The children of a node $w$ are all the walks of the form $e\cdot{}w$ that are nodes in $\Tc$.
        \item The children of a node $w$ are ordered according to the target index of their first edge, that is, if $w_1 = e_1\cdot w$ and $w_2 = e_2\cdot w$ are two distinct children of $w$, then $w_1$ appears before $w_2$ if and only if $\tgtidx(e_1) < \tgtidx(e_2)$. \label{suffixtree:order} 
    \end{enumerate}
    \label{def:suffix-tree}
\end{definition}

\begin{remark}
    The definition of $\Tc$ implies the following properties:
    \begin{itemize}
        \item In item~\ref{suffixtree:order}, it cannot be that $\tgtidx(e_1) = \tgtidx(e_2)$, as $e_1$ and $e_2$ are distinct and have the same target.
        \item All branches of $\Tc$ from the root to a leaf are of length $\lambda$.
        \item The leaves of $\Tc$ are precisely the walks of $\sem{A}(\Dc,s,t)$.
    \end{itemize}
\end{remark}

Intuitively, \functionfont{Enumerate} performs a depth-first traversal of $\Tc$ and outputs precisely the walks that are found at the leaves. The tree will be constructed on the fly, together with a certificate that witnesses that each branch of the tree does indeed correspond to a walk in $\sem{A}(\Dc,s,t)$. This certificate consists in a subset~$S(w)$ of $Q$ attached to each node~$w$ of $\Tc$ that guarantees the existence of at least one accepting run of~$\Ac$ over $w$. It is defined as follows:

\begin{definition}
    For each node $w$ in $\Tc$, we denote by~$S(w)$ the following set.
    \begin{equation*}
        S(w) = \setst*{q\in Q}{\begin{array}{l}
        \exists w_q\in\walks(\Dc)\\
        ~\bullet~(w_q\cdot w)\in \sem{A}(\Dc,s,t) \\
        ~\bullet~q\in\left(\Delta(I,\lbl(w_q))\cap \Delta^{-1}(\lbl(w),F)\right)
        \end{array}}
    \end{equation*}
\end{definition}

Remark that the definition immediately implies that, for a node $w$ in $\Tc$, $S(w) \neq \emptyset$.

Intuitively, $\Tc$ and $S(\cdot)$ can be reconstructed from $C$ as follows. Assume that $\Tc$ and $S(\cdot)$ have already been constructed up to some node $w$. All $p\in S(w)$ are states that can be reached at $w$ and are useful in at least one accepting run that reaches $t$. Then, computing the children of $w$ in $\Tc$ (and $S(\cdot)$ for them) amounts to looking for predecessor edges and states that reach $p$ at $w$. This is precisely the information provided by $C$, as formally stated below.

\begin{lemma}
    \label{lem:c-tree}
    Let~$e\cdot w$ be a node of $\Tc$ for some edge $e$ and walk~$w$. Let~$u = \tgt(e)$. For every~$p$, we denote by~$X_{p}$
    the unique\footnote{There can be at most one, due to Lemma~\ref{lem:trim}, item~\ref{trimlemma:b}.} list of states such that $(e,X_p) \in C_{u}[p]$ if such a list exists,
    or~$X_{p}=\emptyset$ otherwise.
    Then, the following holds.
    \begin{equation*}
        S(e\cdot{}w) = \bigcup_{p\in S(w)} set(X_{p})
    \end{equation*}
\end{lemma}


We can now explain \functionfont{Enumerate} more precisely.
Start from the walk $\walk{t}$ and $S(\walk{t})$ that is explicitly computed in \functionfont{Main} as the set of final states of $\Ac$ that can be reached after a walk of length $\lambda$. 

The goal of \functionfont{Enumerate} is to make a depth-first traversal of $\Tc$ while rebuilding the tree on the fly. When called on a walk $w$, with $u = \src(w)$, and the state set $S(w)$, \functionfont{Enumerate} makes use of $C_u$ to construct the children $e\cdot w$ of $w$ in $\Tc$, along with their certificate $S(e\cdot{}w)$. Indeed, Lemma~\ref{lem:c-tree} states that $S(e\cdot{}w)$ can be built by looking for the pairs $(e,X)$ in each $C_u[p]$ for $p\in S(w)$. We then proceed with the depth-first traversal by starting over with $w = e\cdot w$ and $S = S(e\cdot{}w)$.

There is one technical hurdle to overcome before making a recursive call for an edge $e$: in order to avoid doing two calls for the same $e$, we have to make sure that we collect all occurrences of pairs $(e,X)$ that appear in any $C_u[p]$ for $p\in S(w)$. However, this cannot be done by browsing each $C_u[p]$ entirely, otherwise the delay would depend on (the maximal in-degree of) $\Dc$. This issue is solved by the fact that $C_u[p]$ is sorted in $\tgtidx$ order, as stated in Lemma~\ref{lem:trim}. Thus, to find the pairs $(e,X)$ that correspond to the first child of $w$ in $\Tc$, we only have to search in the head of each $C_u[p]$.

In the pseudocode of \functionfont{Enumerate}, lines~\ref{enumerate:compute min start}-\ref{enumerate:compute min end} look at the head of each $C_u[p]$ to find the minimal edge $e$ that has not yet been found. Then lines~\ref{enumerate:compute X start}-\ref{enumerate:compute X end} correspond to collecting all $(e,X)$ in the head of each $C_u[p]$ for the found edge $e$. Line~\ref{enumerate:return} corresponds to the case where all $C_u[p]$ have been exhausted.

Finally, \functionfont{Enumerate} keeps track, in variable $\ell$, of the remaining distance to the leaves of $\Tc$, and outputs each time it reaches $\ell = 0$.

The correctness of \functionfont{Enumerate} comes from this final lemma:

\begin{lemma}
    The tree of recursive calls to \functionfont{Enumerate} is isomorphic to $\Tc$ in the following sense: \functionfont{Enumerate}($C,\ell, w, S$) is called exactly once per node $w$ in $\Tc$. Moreover, the parameters satisfy~$\ell = \lambda - \len(w)$ and $S = S(w)$.
    \label{lem:enumerate-tree}
\end{lemma}

\section{Complexity analysis}
\label{sec:complexity}

In Section~\ref{sec:time-complexity}, we show that the algorithm meets the complexity bounds claimed in Theorem~\ref{theorem:main} and in
Section~\ref{sec:space-complexity}, we discuss its memory usage.

\subsection{Time complexity of the algorithm}
\label{sec:time-complexity}

\paragraph{Annotate.} Creating and fully initializing a map $B_u$ for some $u$ takes time $\indeg(u) \times \card{Q}$. Thus, creating all the maps takes 
\begin{equation}
    \bigo{\sum_{u\in V} \indeg(u) \times \card{Q}} = \bigo{\card{E}\times \card{Q}}
\end{equation}
The remainder of the initialization is negligible.

The main loop (starting at line \ref{annotate:main loop}) is essentially a breadth-first traversal of~$\Dc\times\Ac$, hence it is not surprising that it runs in $\bigo{\card{E}\times\card{\Delta}}$. Indeed, each pair $(v,q)\in V\times Q$ is visited at most once. Then, for each such $(v,q)$, we visit each outgoing edge $e$ of $v$ (line~\ref{annotate:for(e)}), for which we then visit each outgoing transition of $q$ (line~\ref{annotate:for(p)}). Thus, the elementary instructions, at lines~\ref{annotate:inner loop start}-\ref{annotate:inner loop end}, are executed at most $\card{E}\times\card{\Delta}$ times. Thus, the total runtime of \functionfont{Annotate} is $\bigo{\card{E}\times \card{\Delta}}$.

\paragraph{Trim.} Creating and initializing a map $C_u$ for some $u$ takes time $\bigo{\card{Q}}$. Thus, creating all maps takes time $\bigo{\card{V}\times\card{Q}}$. The remainder of the function consists in three nested for-loops that simply visit each edge-state pair exactly once. Thus, the total runtime of \functionfont{Trim} is $\bigo{\card{E}\times\card{Q}}$.

\paragraph{Enumerate.} As shown in Lemma~\ref{lem:enumerate-tree}, \functionfont{Enumerate} performs a depth-first traversal of $\Tc$ and makes exactly one recursive call per node of the tree. Thus, a leaf of the tree is reached and an output is produced at most every $\lambda$ recursive calls. Within \functionfont{Enumerate}, the time between two consecutive recursive calls (excluding time spent in the recursive calls themselves) is dominated by the two nested for-loops (lines~\ref{enumerate:compute X start}-\ref{enumerate:compute X end}). Lemma~\ref{lem:trim} states that, for a fixed $e$, the size of $(e,X_p) \in C_u[p]$ is bounded by $\sum_{a\in\Sigma}\card{\Delta^{-1}(a,p)}$. Thus, browsing all such $X_p$ for all $p$ in $Q$ takes time $\bigo{\card{\Delta}}$. Hence, \functionfont{Enumerate} produces an output at most every $\bigo{\lambda\times\card{\Delta}}$ steps.

\paragraph{Total.} The preprocessing phase (\functionfont{Annotate} and \functionfont{Trim}) takes time $\bigo{\card{E}\times \card{\Delta}}$ and the delay between two consecutive outputs during \functionfont{Enumerate} is $\bigo{\lambda\times\card{\Delta}}$, which proves Theorem~\ref{theorem:main}.

\subsection{Memory usage}
\label{sec:space-complexity}

Sometimes, a polynomial delay algorithm might end up using exponential space.  Indeed, when the enumeration procedure is allowed to update the precomputed data structure, it could become arbitrarily large after arbitrarily many answers have been outputted. Our algorithm avoids this pitfall.

\begin{remark}
    Throughout the enumeration, the total memory usage of the algorithm never exceeds $\bigo{\card{E}\times\card{\Delta}}$. Remark that the space needed to store the walk of length $\lambda$ that is being outputted is also in $\bigo{\card{E}\times\card{\Delta}}$, and therefore negligible. 
\end{remark}

A stricter class of enumeration algorithms, called \emph{memoryless} \cite{Strozecki2021}, forbids any kind of modification to the precomputed structures during the enumeration. Formally, in a memoryless enumeration algorithm, the $(i+1)$-th output is computed directly from the $i$-th output and the precomputed data structures.

Our algorithm is \emph{not} memoryless, as the efficiency of \functionfont{Enumerate} hinges on reading $C_u$ in the order it was initially created. Thus, enumeration cannot be readily resumed from any given output. We can adapt \functionfont{Enumerate} to the memoryless framework as follows. Given an answer $w$, we perform a computation that is \emph{guided} by~$w$: instead of searching for the minimum edge at lines~\ref{enumerate:compute min start}-\ref{enumerate:compute min end}, we set $e_{min}$ to the last edge $e$ of $w$, advance all queues to the first $e'$ with $\tgtidx(e') \geq \tgtidx(e)$, remove the last edge of $w$ and do the recursive call. Once this computation reaches $\ell = 0$, we skip outputting $w$. Then, all queues have been set to the correct position, and the algorithm resumes as usual to produce the next output. 

Unfortunately, this guided execution costs more than the normal delay of our main algorithm: advancing all queues to match $e$ costs $\bigo{\card{Q}\times\indeg(u)}$, where $u = \tgt(e)$. This leads to an algorithm with memoryless delay in $\bigo{d\times\lambda\times\card{\Delta}}$, where $d$ is the maximum in-degree of $\Dc$. 

The $d$ factor can actually be avoided by using a more involved data structure to represent $C_u$ that allows resuming from a given $e$ in constant time. Formally, $C_u$ will be a copy of $A_u$ in which each cell also contains a pointer to the next non-empty cell. It can be computed in time $\bigo{\card{E}\times\card{Q}}$ as follows:

\begin{figure}[ht]
\begin{vmalgorithm}
    \Function{ResumableTrim}{Annotation $B$}
    \ForAll{$u\in V$}
        \State \parbox{\linewidth}{$C_u \gets$ new map $Q \to \set{0,\ldots,\indeg(u)-1}$}
        \Statex \hfill$\to List[Q\times \mathbb{N}]$
        \ForAll{$p\in Q$}
            \State $next \gets nil$
            \ForAll{$i\in\{\indeg(u)-1,\ldots,0\}$} \smashedcommentpar{Reverse order}
                \State $C_u[p][i] \gets (B_u[p][i],next)$
                \If{$B_u[p][i]$ is not empty}
                    \State $next \gets i$
                \EndIf
            \EndFor
        \EndFor
    \EndFor
    \State \Return {$C$}
\EndFunction
\end{vmalgorithm}
\end{figure}

Then we replace \functionfont{Enumerate} by a new function \functionfont{NextOutput} that acts as its counterpart in the memoryless setting: given the precomputed data structure $C$ and a previous output $w$, the function \functionfont{NextOutput}$(C,w)$ performs a guided run as described previously and then produces the next output. Writing this function poses no conceptual challenge, but requires some technical care while traversing $C$. Altogether, this leads to the following theorem:

\begin{theorem}
    \problemfont{Distinct Shortest Walks}$(\Dc,\Ac,s,t)$ can be enumerated with a \emph{memoryless} algorithm with a preprocessing time in $\bigo{\card{\Dc}\times\card{\Ac}}$ and a delay in $\bigo{\lambda\times \card{\Ac}}$, where~$\lambda$ is the length of a shortest walk.    
\end{theorem}

\section{Extensions}
\label{sec:extensions}

\subsection{Handling spontaneous transitions}
\label{sec:epsilon}

Note that Definition~\ref{def:automaton} disallows $\varepsilon$-transitions. An automaton $\Ac = \aut{\Sigma,Q,\Delta,I,F}$ allows $\varepsilon$-transitions when $\Delta\subseteq Q \times (\Sigma \cup \set{\varepsilon}) \times Q$.

Handling spontaneous transitions requires some light editing of \functionfont{Annotate} and does not impact the other subfunctions. It essentially consists in eliminating $\varepsilon$-transitions \emph{on the fly}: each time a new state $p$ is reached at a vertex $u$, we also add each state $r$ that can be reached from $p$ by following one or more $\varepsilon$-transitions. Formally, it amounts to replacing the innermost loop of \functionfont{Annotate} (lines~\ref{annotate:inner loop start}--\ref{annotate:inner loop end}) with a call to $\textsc{PossiblyVisit}(u, p, e)$, given below.

\begin{figure}[ht]
\begin{vmalgorithm}
\Function{PossiblyVisit}{Vertex~$u\in V$, State~$p \in Q$, Edge~$e \in E$}
    \LComment{We use $\varname{next}$, $\varname{stop}$, $L_u$, $B_u$, $\ell$, $t$ and $\Ac$ from \textsc{Annotate} as global variables}
\If {$p \notin \dom(L_u)$}\label{visit:if p}
     \smashedcommentpar{First time state $p$ is reached at vertex $u$}
     \State $L_u[p] \gets \ell$
     \State add $(u,p)$ to $\varname{next}$
     \If{ $u=t$ ~\textbf{and}~ $p\in F$}                   \smashedcommentpar{First time a final state is reached\\at vertex $t$}
         \State $\varname{stop}\gets \top$ 
     \EndIf
     \State add $q$ to $B_u[p][\tgtidx(e)]$
     \ForAll{$r\in\Delta(p,\varepsilon)$}\label{visit:for each r}
        \State PossiblyVisit($u$, $r$, $e$)
     \EndFor
 \ElsIf {$L_u[p] = \ell$}
     \smashedcommentpar{We found another walk of length $\ell$\\ that reaches state $p$ at vertex~$u$}
     \Statex
     \State add $q$ to $B_u[p][\tgtidx(e)]$ 
 \EndIf
\EndFunction
\end{vmalgorithm}
\label{a:possiblyvisit}
\end{figure}

The test on line \ref{visit:if p} will be true at most once per pair $(u,p)\in Q \times V$. Thus, at the end \functionfont{Annotate}, there will have been at most 
$\card{V} \times \card{\Delta_\varepsilon}$ laps in the for-loop on line \ref{visit:for each r}, where~$\Delta_\varepsilon$ is the set of spontaneous 
transitions in~$\Ac$. Therefore, this modification does not change the time complexity of \functionfont{Annotate}.

\subsection{Query given as a regular expression}

In real-life scenarios, the query is usually given in some query language that is closer to a regular expression than an automaton. For every regular expression $R$, we define $\sem{R}(\Dc,s,t)$ to be equal to $\sem{\Ac}(\Dc,s,t)$ for any automaton $\Ac$ that accepts the language described by $R$. We can now formally define the corresponding computational problem and establish the corresponding complexity bounds.

\begin{problem}{Distinct Shortest Walks (Regexp variant)}
    \item[Inputs:] A multi-labeled multi-edge database~$\Dc$, and two vertices~$s,t$ in~$\Dc$.
    \item[Query:] A regular expression~$R$.
    \item[Output:] Enumerate~$\sem{R}(D,s,t)$, without duplicates.
\end{problem}

The usual approach for handling a regular expression $R$ consists in translating it first to some equivalent automaton $\Ac$. We know from Section~\ref{sec:epsilon} that our algorithm works on automata with \mbox{$\varepsilon$-transitions} at no additional cost. Thus, we can readily use Thompson construction (Theorem~\ref{theorem:thompson} below) which, combined with Theorem~\ref{theorem:main}, immediately leads to Corollary~\ref{corollary:main-regexp}, stated afterwards.

\begin{theorem}[\citeauthor{Thompson1968}, \citeyear{Thompson1968}]\label{theorem:thompson}
    Given a regular expression~$R$, there is an algorithm that runs in time $\bigo{\card{R}}$
    and build an equivalent automaton with \mbox{$\varepsilon$-transitions} $\Ac$ with $\bigo{\card{R}}$ states and $\bigo{\card{R}}$ transitions in total.
\end{theorem}

\begin{corollary}\label{corollary:main-regexp}
    When the query is given as a regular expression $R$, \problemfont{Distinct Shortest Walks} can be enumerated with a preprocessing time in $\bigo{\card{R}\times\card{\Dc}}$ and a delay in $\bigo{\lambda \times\card{R}}$, where~$\lambda$ is the length of a shortest walk.
\end{corollary}

    A more common translation from regular expressions to automata consists in using Glushkov construction \cite{BrugemannKlein93}.
        The produced NFA has no $\varepsilon$-transitions, but may have up to
    $\bigo{\card[2]{R}}$ transitions. In our case, this would yield weaker complexity bounds: $\bigo{\card[2]{R}\times\card{\Dc}}$ preprocessing time and $\bigo{\lambda \times \card[2]{R}}$ delay.

\subsection{Adaptation to related problems}

In this section, we discuss several problems that are related to \problemfont{Distinct Shortest Walks} and briefly explain how our algorithm can be adapted to address them.

\paragraph{One source to many targets.} In \problemfont{Distinct Shortest Walks}, both source and target vertices $s$ and $t$ are given as part of the input. Another version of the problem only fixes $s$, and then asks for shortest walks from $s$ to $t$ for a subset of (or possibly all) the vertices $t$ in $\Dc$. This problem can be solved at no additional cost: start by running \functionfont{Annotate} and add wanted targets to a queue the first time they are annotated with a final state of $\Ac$. This step stops when no new pair $(v,q)$ can be discovered, and thus has the same worst-case complexity as the main algorithm. Then, for each queued target, collect the set of final states that have been reached with a walk of minimal length, and run the enumeration step as usual.

\paragraph{Distinct Cheapest Walks.} In this scenario, in addition to their labels, edges of $\Dc$ carry a positive value, called \emph{cost}. Instead of looking for shortest walks from $s$ to $t$, this problem asks for \emph{cheapest} walks, that is, walks that minimize the sum of costs along their edges. Our algorithm can be easily adapted to this setting by replacing the breadth-first traversal in \functionfont{Annotate} with a cheapest-first traversal, as is done in Dijkstra's algorithm. In that case, the preprocessing time complexity becomes 
 \begin{equation*}
     \bigo{\card{\Dc}\times\card{\Ac} + \card{V}\times\card{Q}\times\big(\log(\card{V})+\log(\card{Q})\big )}
 \end{equation*}
 using standard techniques (\cite{Fredman1984}) and the delay is unchanged.

\paragraph{Shortest Walks with Multiplicities.} This version of the problem asks to return all shortest walks $w$ from $s$ to $t$ together with their \emph{multiplicity}, that is, the number of different accepting runs of $\Ac$ over $\lbl(w)$. Theoretically, one could rerun $\Ac$ on $w$ when it is output, and simply count the runs. Indeed, this would cost $\bigo{\lambda \times \card{\Ac}}$, and would not change the delay. That being said, our algorithm essentially runs $\Ac$ over $w$ along the recursive calls to \functionfont{Enumerate}. Hence, it can be easily adapted to keep track of the number of times each state has been produced along the walk.

\section{Perspectives}
\label{sec:perspectives}

We have proposed an algorithm for solving \problemfont{Distinct Shortest Walks} that achieves a  $\bigo{\lambda\times\card{\Ac}}$ delay after a preprocessing time in $\bigo{\card{\Dc}\times\card{\Ac}}$. In this final section, we briefly discuss lower bounds and potential leads to improve our upper bounds.

%
It is unlikely that the preprocessing time of our algorithm can be improved by a polynomial factor.
Indeed, several results \cite{EquiMakinen+2023, BackursIndyk2016} show that, under the Strong Exponential Time Hypothesis (SETH), deciding whether a word $w$ matches a regular expression~$R$ cannot be done in $\bigo{\card[1-\varepsilon]{w}\times\card{R}}$
nor $\bigo{\card{w}\times\card[1-\varepsilon]{R}}$ for any~$\varepsilon>0$.
Deciding whether \problemfont{Distinct Shortest Walk} has at least one output subsumes this problem. In other words, under SETH, the preprocessing time of \problemfont{Distinct Shortest Walk} cannot belong to $\bigo{\card[1-\varepsilon]{\Dc}\times\card{R}}$ nor  $\bigo{\card{\Dc}\times\card[1-\varepsilon]{R}}$, for any~$\varepsilon>0$.
However, it might be possible to reduce it by a polylogarithmic factor. Indeed, Myers \cite{Myers1992} gave an algorithm in $\bigo{\frac{\card{R}\times\card{w}}{\log \card{w}}+\card{w}}$ for the former problem. It was later improved by \citeauthor{BilleThorup2009} \cite{BilleThorup2009} to run in $\bigo{\frac{\card{R}\times\card{w}}{\log^{1.5} \card{w}}+\card{w}}$.
These results provide interesting leads for improving our algorithm.


A significant part in the delay comes from the time taken to actually write down the output in full. However, it is likely that most walks have large parts in common, especially if the set of answers is larger than the size of the database. In that case, one may significantly decrease the delay by 
outputting only the difference with the previous output.  In that case, the order in which the walks are produced is crucial. In a recent article \cite{AmariliMonet2023}, Amarilli and Monet showed how to find efficient orders for enumerating a regular language given as an automaton. Using similar techniques might allow to significantly reduce the (amortized) delay of our algorithm.




\bibliographystyle{ACM-Reference-Format}
\bibliography{bibliography.bib, languages.bib}


\ifincludeappendix
    \appendix\onecolumn\Large
\newcounter{appendix}\setcounter{appendix}{0}
\renewcommand{\theappendix}{\Alph{appendix}}
\renewcommand{\thesection}{\theappendix.\arabic{section}}
\renewcommand{\thetheorem}{\theappendix.\arabic{theorem}}
\renewcommand{\theproposition}{\theappendix.\arabic{proposition}}
\renewcommand{\thecorollary}{\theappendix.\arabic{corollary}}
\renewcommand{\thelemma}{\theappendix.\arabic{lemma}}
\renewcommand{\thedefinition}{\theappendix.\arabic{definition}}

\newcommand{\newappendix}[1]{%
  \refstepcounter{appendix}%
  \setcounter{section}{0}%
  \clearpage%
  {\centering%
  \Huge\textbf{Appendix \theappendix:\quad #1}}%
  \vspace*{1cm}
}

\newappendix{Reduction from \problemfont{Distinct Shortest Walks} to \problemfont{All Shortest Words}}
\label{appendix:reduction-to-ackerman-shallit}

\newcommand{\minarrow}{{\normalfont\textsf{MinArrow}}}
\newcommand{\comp}{{\normalfont\textsf{Comp}}}
\newcommand{\lex}{\mathrel{\leq_{LEX}}}

Appendix~\theappendix{} shows how existing work can be used to solve \problemfont{Distinct Shortest Walks}, albeit with a worse complexity than what we achieve in this paper. 
More precisely, we explain how the problem reduces to \problemfont{All Shortest Words}, a tasks which consists in enumerating all words of minimal length that are accepted by a given NFA in lexicographic order.

Formally, \problemfont{All Shortest Words} is defined as follows:

\begin{problem}{All Shortest Words}
    \item[Input:] A nondeterministic automaton~$\Ac=\aut{\Sigma,Q,\Delta,I,F}$.
    \item[Output:] Enumerate all shortest words in~$L(\Ac)$, without duplicates and in lexicographic order.
\end{problem}

The best known algorithm to solve \problemfont{All Shortest Words} is given in \cite{AckermanShallit2009}. However, \citeauthor{AckermanShallit2009} did not explicitly write their algorithm in the enumeration complexity framework. This is done in the techincal report \cite{FrancisMarsault2023}, which proves the following theorem:

\begin{theorem}
    \label{t:ackerman-shallit}
    \problemfont{All Shortest Words} can be enumerated with $\bigo{\lambda\times\card{\Delta} + \lambda\times\card[2]{Q}}$ preprocessing and $\bigo{\lambda\times\card{\Delta}}$ delay, where $\lambda$ is the length of any shortest word in $L(\Ac)$.
\end{theorem}

We show how \problemfont{Distinct Shortest Walks} reduces to \problemfont{All Shortest Words}. This reduction follows very closely the ideas in \cite{MartensTrautner2018}. However, the authors were only interested in showing that the problem can be enumerated with polynomial delay, and hence did not look at the fine-grained complexity. Moreover, our data model slightly differs from theirs. Thus, for the sake of completeness, we chose to rewrite their proof directly on our model, and then give the corresponding complexity. This leads to Theorem~\ref{t:Martens-Trautner} given in the introduction:

\begin{falsestatement}{Theorem~\ref{t:Martens-Trautner}}[\citeauthor{MartensTrautner2018}, \citeyear{MartensTrautner2018}]
Given a nondeterministic automaton $\Ac$ with set of states $Q$ and transition table $\Delta$ and a database $\Dc$ with set of vertices $V$, \problemfont{Distinct Shortest Walks}($\Dc$,$\Ac$) can be enumerated with delay in $\bigo{\card{\Dc}\times\card{\Delta}
\times\lambda}$ after a preprocessing in $\bigo{\card[2]{Q}\times \card[2]{V}\times\lambda+\card{\Delta}\times\card{\Dc}\times\lambda}$, where $\lambda$ is the length of a shortest walk.
\end{falsestatement}

Let $\Dc = (\Sigma, V, E,\src,\tgt,\lbl)$ be a database, $\Ac=\aut{\Sigma,Q,\Delta,I,F}$ an automaton, and $s,t$ two vertices in $\Dc$.

We define a new automaton $\Ac' = \aut{\Sigma',Q',\Delta',I',F'}$ as follows:
\begin{itemize}
    \item $\Sigma' = E$;
    \item $Q' = V\times Q$;
    \item $\Delta' = \left \{(v_1,q_1),e,(v_2,q_2) \ \middle | \
        \begin{array}{l}
            \src(e) = v_1, \tgt(e) = v_2 \\
            \exists a \in \lbl(e), (q_1,a,q_2)\in\Delta
        \end{array}
    \right \}$

    \item $I' = \set{s}\times I$
    \item $F' = \set{t}\times F$
\end{itemize}

We conclude the reduction by remarking that there is a one-to-one mapping from words in $L(\Ac')$ to walks from $s$ to $t$ that match $\Ac$. After each output of the algorithm for \problemfont{All Shortest Words}, we need to reconstruct the corresponding path. This is done in time $\bigo{\lambda}$ by simply retrieving the source and target vertices of each edge. Hence, this is negligible when compared to the delay, which is $\bigo{\lambda\times\card{Delta'}} = \bigo{\lambda\times\card{E}\times\card{\Delta}}$.

Constructing $\Ac'$ takes time $\bigo{\card{V}\times\card{Q} + \card{\Delta}\times\card{E}}$. Again, this is negligible when compared to the delay of \problemfont{All Shortest Words}, in $\bigo{\lambda\times\Delta' + \lambda\times\card[2]{Q'}}$. Altogether, the delay is in $\bigo{\lambda\times\card{\Delta}\times\card{E} + \lambda\times\card[2]{V}\times\card[2]{Q}}$.

\newappendix{Proofs}
\label{appendix:proofs}

Appendix~\ref{appendix:proofs} contains the proofs of all lemmas in the body of the article.

\medskip

Recall that, in all statements, we have a fixed automaton $\Ac = \aut{\Sigma,Q,\Delta,I,F}$, a database $\Dc = (\Sigma, V, E,\src,\tgt,\lbl)$ and two vertices $s$ and $t$ of $\Dc$. Moreover, we assume that there exists at least one walk from $s$ to $t$ that matches $\Ac$, and let $\lambda$ denote the minimal length of any such walk.

\section{Proof of Lemma~\ref{lem:annotate}}

\begin{falsestatement}{Lemma~\ref{lem:annotate}} 
At the end of \functionfont{Annotate}, the following properties hold for all $p,q\in Q$, $u\in V$ and $e\in E$:
    \begin{enumerate}
        \item $L_u[p] = \min(\set{\len(w) \ | \ w\in\walks^{\leq\lambda}(\Dc), s\xrightarrow{w}u \text{ and } p\in\Delta(I,\lbl(w))})$

        \item $q\in{}B_u[p][i]$ if and only if there exists a walk $w = w'\cdot{}e$ from $s$ to $u$ such that:
            \begin{itemize}
                \item $\len(w) = L_u[p]$
                \item $\tgtidx(e) = i$
                \item $q\in\Delta(I,\lbl(w'))$
                \item $p\in\Delta(q,\lbl(e))$
            \end{itemize}

        \item $B_u[p][\tgtidx(e)]$ is of size at most $\sum_{a\in\Sigma}\card{\Delta^{-1}(a,p)}$.
    \end{enumerate}    
\end{falsestatement}

\begin{proof}
    The proof of \ref{annotatelemma:c} comes from the fact that each pair $(v,q)$ is explored at most once during \functionfont{Annotate} at line~\ref{annotate:for(v,q)}. Thus, for $e\in\outgoing(v)$ and $a\in\lbl(e)$, each element $p \in \Delta(q,a)$ adds at most one item in $B_u[p][\tgtidx(e)]$ at line~\ref{annotate:update1} or line~\ref{annotate:update2}, which yields the required bound.

    \medskip

    The proof of \ref{annotatelemma:a} and \ref{annotatelemma:b} are more involved, and require a stronger auxiliary statement.

    For $\ell >0$, let $L_u^{(\ell)}$, $B_u^{(\ell)}$ and $\text{next}^{(\ell)}$ respectively denote the contents of $L_u$, $B_u$ and $\text{next}$ at the end of the $\ell$-th step of \textsc{Annotate}, with $L_u^{(0)}$, $B_u^{(0)}$ and $\text{next}^{(0)}$ being their respective contents after the initialization, at line~\ref{annotate:main loop}.

    We prove the following properties by induction over $\ell \leq \lambda$, which immediately imply the required properties when applied with $\ell = \lambda$.

    \begin{enumerate}
        \item $L_u^{(\ell)}[p] = \min(\set{\len(w) \ | \ w\in\walks^{\leq \ell}(\Dc), s\xrightarrow{w}u \text{ and } p\in\Delta(I,\lbl(w))})$

        \item $q \in{}B_u^{(\ell)}[p][i]$ if and only if there exists a walk $w = w'\cdot e$ from $s$ to $u$ such that:
            \begin{itemize}
                \item $\len(w) = L_u^{(\ell)}[p]$
                \item $\tgtidx(e) = i$
                \item $q\in\Delta(I,\lbl(w'))$
                \item $p\in\Delta(q,\lbl(e))$
            \end{itemize}

        \item[($\star$)] $(u,p)\in\text{next}^{(\ell)}$ if and only if $L_u^{(\ell)}[p] = i$
    \end{enumerate}

    \medskip

    {\bf Initialization step:} it is immediate that the three properties hold for $\ell = 0$, since $B_u^{(0)}$ is empty for each vertex $u$ and $L_s^{(0)}$ correctly reflects the fact that the only walk of length $0$ starting
    at $s$ has an empty label and ends in $s$.

    \medskip

    {\bf Induction step:} assume that the three properties hold for some $\ell\geq 0$ and that there is an $(\ell+1)$-th step. In other words, the algorithm did not end at step $\ell$, ie. $\ell < \lambda$.
    \begin{itemize}
        \item[$\Rightarrow$:] (\ref{annotatelemma:a}) Assume that $L_u^{(\ell+1)}[p] = k$ for some $u$, $p$ and $k$. Remark that once a key-value pair is added to $L_u$, the algorithm never replaces or removes it. Thus, there are two cases:
        \begin{itemize}
            \item Case 1: at step $\ell$, we already have $L_u^{(\ell)}[p] = k$. In that case, the induction hypothesis immediately gives the desired property for $L_u^{(\ell+1)}[p] = k$.
            \item Case 2: at step $\ell$, $L_u^{(\ell)}[p]$ is undefined. In that case, it is defined during step $\ell+1$, after the test at line~\ref{annotate:inner loop start} and $k=\ell+1$. At this point in the algorithm, we know that there exist $(v,q)\in\text{next}^{(\ell)}$ and $e\in\outgoing(v)$ such that $u = \tgt(e)$ and $p\in\Delta(q,\lbl(e))$.

            Since $(v,q)\in\text{next}^{(\ell)}$, the induction hypothesis implies that $L_v[q] = \ell$. Thus there exists a walk $w$ going from $s$ to $v$ with $q\in\Delta(I,\lbl(w))$. Moreover, $w$ is of length $\ell$, which is minimal among all walks with the same endpoints that can reach state $q$ at $v$.

            Thus, $w\cdot e$ is a walk from $s$ to $u$ with $p\in\Delta(I,\lbl(w\cdot e))$. It is of length $\ell+1$, which is indeed minimal, otherwise the induction hypothesis would not allow $L_u^{(\ell)}[p]$ to be undefined.
        \end{itemize}

        \medskip

        (\ref{annotatelemma:b}) Similarly, assume that $q\in B_u^{(\ell+1)}[p][\tgtidx(e)]$ for some $u$, $p$, $q$ and $e$. Then, either it was already true at step $\ell$ and the induction hypothesis immediately gives the result, or $q$ was appended during step $\ell+1$, at line~\ref{annotate:update1} or at line~\ref{annotate:update2}. In both cases, this only happens if $L_u^{(\ell+1)}[p]$ is defined, so that the previous reasoning applies and yields the required walk.

        \medskip

        $(\star)$ Finally, assume that $(u,p)\in\text{next}^{(\ell+1)}$. Once again, this can only happen when $L_u^{(\ell+1)}[p]$ is defined, and the previous reasoning yields $L_u^{(\ell+1)}[p] = \ell+1$.

        \item[$\Leftarrow$:] Assume that there exists a walk $w \in \walks^{\leq \ell+1}(\Dc)$ such that $s \xrightarrow{w} u$ and $p\in\Delta(I,\lbl(w))$ for some vertex $u$ and state $p$. Moreover, assume that $w$ is of minimal length $k$ among such walks.

        \begin{itemize}
            \item Case 1: $k < \ell+1$. Then $w \in \walks^{\leq \ell}(\Dc)$. In that case, the induction hypothesis yields $L_u^{(\ell)}[p] = k$. Since the algorithm never removes a key-value pair in $L$, we immediately get $L_u^{(\ell+1)}[p] = k$. Similarly, if the conditions of (\ref{annotatelemma:b}) are satisfied for $w$ and some $w'$, $e$ and $p$, then the induction hypothesis yields $q\in B_u^{(\ell)}[p][\tgtidx(e)]$, from which we get $q\in B_u^{(\ell+1)}[p][\tgtidx(e)]$. Finally, there is nothing to prove for $\text{next}^{(\ell+1)}$, since $k < \ell+1$.

            \item Case 2: $k = \ell+1$. In that case, there exists $w'$, a vertex $v$ and an edge $e\in\outgoing(v)$ such that $w = s\xrightarrow{w'} v \xrightarrow{e} u$. Since $p \in \Delta(I, \lbl(w))$, there exists $q\in\Delta(I, \lbl(w'))$ such that $p\in \Delta(q, \lbl(e))$. 
            
            Remark that $w'$ is necessarily of minimal length $\ell$ among walks going from $s$ to $v$ that can reach state $q$, otherwise $w$ would not be of minimal length. Thus, the induction hypothesis yields $L_v^{(\ell)}[q] = \ell$ and $(v,q) \in \text{next}^{(\ell)}$. This means that, at step $\ell+1$ of \textsc{Annotate}, $(v,q)$ is added to $\text{current}$. It simply remains to check that $e$ and $p$ satisfy all the conditions so that $L_u^{(\ell+1)}[p] = \ell+1$, $q\in B_u^{(\ell+1)}[p][\tgtidx(e)]$ and $(u,p)$ is added to $\text{next}^{(\ell+1)}$. The only hurdle is to prove that, during this step, either $L_u[p] = \ell+1$ or $L_u[p]$ is undefined, otherwise the induction hypothesis would once again contradict the minimality of $w$. \qedhere
        \end{itemize}
    \end{itemize}
\end{proof}

\section{Proof of Lemma~\ref{lem:trim}}

\begin{falsestatement}{Lemma~\ref{lem:trim}}
    At the end of \functionfont{Trim}, the following properties hold for all $p\in Q$, 
    $u\in V$, $e,e'\in \incoming(u)$ and lists $X,X'$ over $Q$:
    \begin{enumerate}
        \item $(e,X) \in C_u[p]$ if and only if $X = B_u[p][\tgtidx(e)]$ and $X \neq \emptyset$.
        \label{trimlemma:a}
        \item $C_u[p]$ is sorted in increasing $\tgtidx(e)$ order, that is, if $(e,X)$ appears before $(e',X')$ in $C_u[p]$, then $\tgtidx(e) < \tgtidx(e')$. In particular, $e \neq e'$.
        \label{trimlemma:b}
        \item If $(e,X) \in C_u[p]$, then $X$ is of size at most $\sum_{a\in\Sigma}\card{\Delta^{-1}(a,p)}$.
        \label{trimlemma:c}
    \end{enumerate}
\end{falsestatement}

\begin{proof}
    The proof of (\ref{trimlemma:a}) immediately follows from the pseudocode of \functionfont{Trim}. Indeed, \functionfont{Trim} browses all $B_u$ exhaustively\footnote{Indeed, $B_u[p]$ ranges over $\incoming(u)$} and enqueues precisely the pairs $(e,B_u[p][\tgtidx(e)])$ for which $B_u[p][\tgtidx(e)]$ is not empty.

    Moreover, $B_u$ is explored in the same order as $\incoming(u)$ (at line~\ref{trim:for(e)}), thus pairs $(e,X)$ are enqueued in $C_u$ in the same order as $e$ appears in $\incoming(u)$. Thus, (\ref{trimlemma:b}) follows from the definition of $\tgtidx(e)$, which is precisely the position where $e$ appears in $\incoming(u)$.

    Finally, (\ref{trimlemma:c}) immediately follows from (\ref{trimlemma:a}) together with Lemma~\ref{lem:annotate}, item \ref{annotatelemma:c}.
\end{proof}

\section{Proof of Lemma~\ref{lem:c-tree}}

\begin{falsestatement}{Lemma~\ref{lem:c-tree}}
    Let~$e\cdot w$ be a node of $\Tc$ for some edge $e$ and walk $w$. Let $u = \tgt(e)$. For every~$p$, we denote by~$X_{p}$
    the unique list of states such that $(e,X_p) \in C_{u}[p]$ if such a list exists,
    or~$X_{p}=\emptyset$ otherwise.
    Then, the following holds.
    \begin{equation*}
        S(e\cdot{}w) = \bigcup_{p\in S(w)} set(X_{p})
    \end{equation*}
\end{falsestatement}

\begin{proof}
    Let $e,u,w$ be defined as in the statement of the lemma. Let $v = \src(e)$.

    \medskip

    $\subseteq:$ Let $q\in S(e\cdot{}w)$. By definition of $S$, there exists a walk $w_q$ such that $w_q\cdot{}e\cdot{}w \in \sem{\Ac}(\Dc,s,t)$, with $q\in\Delta(I,\lbl(w_q))$ and $q\in\Delta^{-1}(\lbl(e\cdot{}w),F)$.

    Remark that $w_q$ is a walk from $s$ to $v$ with $q\in\Delta(I,\lbl(w_q))$. Moreover, $w_q$ is of minimal length among such walks, otherwise $w_q\cdot{}e\cdot{}w \notin \sem{\Ac}(\Dc,s,t)$.

    We know $\Delta(q,\lbl(e)) \neq \emptyset$, otherwise we could not have $q\in\Delta^{-1}(\lbl(e\cdot{}w),F)$. Thus, let $p\in\Delta(q,\lbl(e))$. Then $w_q\cdot{}e$ is a walk from $s$ to $u$ with $p\in\Delta(I,\lbl(w_q))$, and it is of minimal length among such walks, for similar reasons.

    We can apply Lemma~\ref{lem:annotate} to $w_q\cdot{}e$. Thus, $q\in{}B_u[p][\tgtidx(e)]$.

    Then, Lemma~\ref{lem:trim} provides a set $X_p$ such that $q\in X_p$ and $(e,X_p) \in C_u[p]$.

    It remains to remark that $p\in S(w)$. Indeed, in the definition of $S$, we can choose the walk $w_p = w_q\cdot e$ as a witness.

    \medskip

    $\supseteq:$ Let $q\in X_p$ for some $p\in S(w)$. 
    
    Since $p\in S(w)$, it means that there exists a walk $s\xrightarrow{w_p}u$ that reaches $p$ at $u$, such that $w_p\cdot w\in\sem{\Ac}(\Dc,s,t)$, with $p\in\Delta^{-1}(\lbl(w),F)$. As for the direct inclusion, this implies that $w_p$ is of minimal length among such walks.
    
    Since $q\in X_p$, Lemma~\ref{lem:trim} implies that $q\in B_u[p][\tgtidx(e)]$.

    Thus, Lemma~\ref{lem:annotate} shows that there exists a walk $w_2 = w'_2\cdot{}e$ from $s$ to $u$ with $q\in\Delta(I,\lbl(w'_2))$ and $p\in\Delta(q,\lbl(e))$. Moreover, $w_2$ is of minimal length among walks that reach $p$ at $u$. Since $p\in\Delta^{-1}(\lbl(w),F)$, we deduce that $w_2\cdot w$ reaches a final state at $t$. Additionally, $\len(w_p) = \len(w_2)$ (as they both have minimal length), thus $w_2\cdot w\in \sem{\Ac}(\Dc,s,t)$.

    Then, we simply remark that $w'_2$ is a suitable witness to show that $q\in S(e\cdot w)$ in the definition of $S$.
\end{proof}

\section{Proof of Lemma~\ref{lem:enumerate-tree}}

The proof of Lemma~\ref{lem:enumerate-tree} requires an additional result:

\begin{lemma}
    Let $w_1$ be a node of $\Tc$ and $w_2$ be a strict descendant of $w_1$ such that $\src(w_1) = \src(w_2)$. Then $S(w_1) \cap S(w_2) = \emptyset$.
    \label{lem:noconcurrency}
\end{lemma}

\begin{proof}
    Let $w_1$ and $w_2$ be defined as in the statement of the lemma. Since $w_2$ is a descendant of $w_1$, by definition of $\Tc$, there exists a walk $w_2'$ with $\len(w_2') \geq 1$ such that $w_2 = w_2'.w_1$
    
    By contradiction, assume that there exists $q\in S(w_1) \cap S(w_2)$. Then, we know that there exists two walks $w_{1q}$ and $w_{2q}$ such that:
    \begin{itemize}
        \item $w_{1q}\cdot w_1 \in \sem{\Ac}(\Dc,s,t)$, with $q\in \Delta(I,\lbl(w_{1q})$ and $q\in \Delta^{-1}(\lbl(w_1),F)$.
        \item $w_{2q}\cdot w_2 \in \sem{\Ac}(\Dc,s,t)$, with $q\in \Delta(I,\lbl(w_{2q})$ and $q\in \Delta^{-1}(\lbl(w_2),F)$.
    \end{itemize}

    Now, remark that $w_{2q}$ and $w_1$ concatenate. Indeed, $\tgt(w_{2q}) = \src(w_2) = \src(w_1)$. Thus, $s\xrightarrow{w_{2q}\cdot w_1}t$. Moreover, $w_{2q}\cdot w_1$ reaches a final state at $t$, because $q\in\Delta(I,\lbl(w_2))$ and $q\in\Delta^{-1}(\lbl(w_1),F)$, which means that $w_{2q}\cdot w_1$ matches $\Ac$. 
    
    However, $w_{2q}\cdot w_1$ is shorter than $w_{2q}\cdot w_2$. Indeed, $\len(w_{2q}\cdot w_2) = \len(w_{2q}\cdot w_1) + \len(w_2')$ and $\len(w_2') \geq 1$. This is a contradiction with $w_{2q}\cdot w_2\in \sem{\Ac}(\Dc,s,t)$.
\end{proof}

We are now ready to prove Lemma~\ref{lem:enumerate-tree}.

\begin{falsestatement}{Lemma~\ref{lem:enumerate-tree}}
    The tree of recursive calls to \functionfont{Enumerate} is isomorphic to $\Tc$ in the following sense: \functionfont{Enumerate}($C,\ell, w, S$) is called exactly once per node $w$ in $\Tc$. Moreover, the parameters of this call satisfy $\ell = \lambda - \len(w)$ and $S = S(w)$. 
\end{falsestatement}

\begin{proof}
    We prove the result by induction over the tree of recursive calls to \functionfont{Enumerate}.

    \textbf{Initialization step:} The first call, in \functionfont{Main}, is 
    $\functionfont{Enumerate}(C,\lambda,\walk{t},S_t)$, where $S_t = \setst{q}{L_t[q] = \lambda} \cap F$. As stated in the lemma, the walk $\walk{t}$ corresponds to the root of $\Tc$ and we have $\lambda = \lambda - \len(\walk{t})$, since $\walk{t}$ is a walk of length 0. It remains to show that $S_t = S(\walk{t})$. Indeed, we have the following equivalences:
    
    \begin{align*}
        q\in S_t &\Leftrightarrow L_t[q] \mathbin= \lambda \textrm{ and } q\in F \\
        \textrm{(Lemma~\ref{lem:annotate})}    &\Leftrightarrow \exists w_q, s\mathbin{\xrightarrow{w_q}}t, q\in \Delta(I,\lbl(w_q)) \cap F \textrm{ and } \len(w_q) \mathbin= \lambda \\
            &\Leftrightarrow \exists w_q\in\sem{\Ac}(\Dc,s,t)\textrm{ and }q\in \Delta(I,\lbl(w_q)) \cap F \\
        \textrm{(with $w_q \mathbin= w_q\cdot \walk{t})$}    &\Leftrightarrow q \in S(\walk{t})
    \end{align*}

    \textbf{Induction step:} Assume that the property holds for $\functionfont{Enumerate}(C,\ell,w,S)$ and all its ancestors in the tree of recursive calls of \functionfont{Enumerate}. We now have to show that it holds for its recursive calls.

    \begin{itemize}
        \item Case 1: $\ell = 0$. In that case, \functionfont{Enumerate} stops without making recursive calls, at line~\ref{enumerate:base}. It remains to show that $w$ has no child in $\Tc$. Indeed, since $\ell = 0$, the induction hypothesis yields $\len(w) = \lambda$ and we know from the definition of $\Tc$ that the nodes $w$ of length $\lambda$ are precisely at the leaves.

        \item Case 2: $\ell > 0$. This case requires some care, as the same structure $C$ is shared between all calls to \functionfont{Enumerate}. Hence, we first have to make sure that previous or concurrent calls will not interfere with the execution of the current call. This is stated in the following claim:

        \begin{falsestatement}{Claim} Calls to \functionfont{Enumerate} do not make concurrent access to the same data structures, in the following sense:
            \begin{itemize}
                \item At the beginning of a call to \functionfont{Enumerate}, all queues $C_u[p]$ that will be read during this call are on their starting position.
                \item If a call to \functionfont{Enumerate} reads or advances a queue $C_u[p]$, then none of its ancestors reads nor advances the same queue.
                \item At the end of a call to \functionfont{Enumerate}, all queues $C_u[p]$ that have been advanced during this call have been restarted.
            \end{itemize}
        \end{falsestatement}

        This claim (up to the current call) immediately follows from the induction hypothesis together with Lemma~\ref{lem:noconcurrency} and the fact that \functionfont{Enumerate} restarts used queues before returning, on line~\ref{enumerate:restart}.

        We can now reason about the current call $\functionfont{Enumerate}(C,\ell,w,S)$. Let $u = \src(w)$, as set at line~\ref{enumerate:beginning}. First, the induction hypothesis yields $S = S(w)$. Thus, from the claim, we deduce that, at the beginning, all queues $C_u[p]$ for $p\in S(w)$ are on their starting position. Hence, the loop at lines~\ref{enumerate:compute min start}-\ref{enumerate:compute min end} computes $e_{min}$ as the least $e$ (in $\tgtidx$ order) such that $(e,X)\in C_u[p]$ for some $X$ and $p\in S(w)$. Indeed, we know from Lemma~\ref{lem:trim} that $C_u[p]$ is sorted, hence the least $e$ can only appear in the head of the queues. Additionally, it cannot be that all queues are empty. Indeed, $w$ must have a child in $\Tc$ (otherwise $\len(w) = \lambda$ and $\ell = 0$). Thus, $S(w') \neq \emptyset$ and Lemma~\ref{lem:c-tree} ensures that at least one queue is not empty.

        Next, the loop at lines~\ref{enumerate:compute X start}-\ref{enumerate:compute X end} computes the union of all $X$ such that $(e_{min},X)\in C_u[p]$ for some $p\in S(w)$, once again thanks to the fact that $C_u[p]$ is sorted. From Lemma~\ref{lem:c-tree}, we know that this is precisely $S(e_{min} \cdot w)$. Thus, the parameters of the first recursive call at line~\ref{enumerate:enumerate} correctly correspond to the first child of $w$ in~$\Tc$.

        For the subsequent calls, simply remark that the loop only advanced the queues that had $e_{min}$ in the head. Thus, we can repeat the same reasoning when $e_{min}$ finds the second least element in $C_u[p]$ for $p\in S(w)$, and so on, until all queues are exhausted. In the end, $e_{min} = nil$, \functionfont{Enumerate} restarts all used queues and returns. \qedhere
    \end{itemize}
\end{proof}
\fi

\end{document}